\documentclass[preprint,review,1p,onecolumn]{elsarticle}

\usepackage{latexsym,amssymb,amsmath}
\usepackage{enumerate}

\usepackage{epsfig}
\usepackage[nothing]{algorithm}
\usepackage[noend]{algpseudocode}

\usepackage{multirow}
\usepackage{array}

\newcommand{\First}[0]{{\mbox{\sf First}}}
\newcommand{\Last}[0]{{\mbox{\sf Last}}}
\newcommand{\Z}[0]{{\mbox{\sf Z}}}

\newcommand{\A}{{\underline{a}}}
\newcommand{\C}{{\underline{c}}}
\newcommand{\T}{{\underline{t}}}

\newcommand{\partitle}[1]{}                        
\newcommand{\commentout}[1]{}

\newtheorem{lemma}{Lemma}
\newtheorem{theorem}{Theorem}

\newdefinition{definition}{Definition}

\newenvironment{proof}{\noindent{\bf Proof:\rm}}{\qed \bigbreak}

\newcommand{\bi}{\begin{itemize}}
\newcommand{\ei}{\end{itemize}}

\newcommand{\be}{\begin{enumerate}}
\newcommand{\ee}{\end{enumerate}}

\newcommand{\bd}{\begin{description}}
\newcommand{\ed}{\end{description}}

\journal{}

\begin{document}

\begin{frontmatter}





\title{FM-index of Alignment with Gaps}

\author[Sejong]{Joong Chae Na}

    \ead{jcna@sejong.ac.kr}

    \address[Sejong]{Department of Computer Science and Engineering,
        Sejong University,
        Seoul 05006, South Korea}

\author[SNU]{Hyunjoon Kim}

    \ead{hjkim@theory.snu.ac.kr}

    \address[SNU]{School of Computer Science and Engineering,
        Seoul National University,
        Seoul 08826, South Korea}

\author[SNU]{Seunghwan Min}

    \ead{shmin@theory.snu.ac.kr}

\author[Hanyang]{Heejin Park}

    \ead{hjpark@hanyang.ac.kr}

    \address[Hanyang]{Department of Computer Science and Engineering,
        Hanyang University,
        Seoul 04763, South Korea}

\author[Rouen,Perth]{Thierry Lecroq}

    \ead{Thierry.Lecroq@univ-rouen.fr}

    \address[Rouen]{Normandie University,
         UNIROUEN, UNIHAVRE,INSA Rouen, LITIS,
         76000 Rouen, France}
    \address[Perth]{Centre for Combinatorics on Words \& Applications,
        School of Engineering \& Information Technology, Murdoch University,
        Murdoch WA 6150, Australia}

\author[Rouen]{Martine L\'eonard}

    \ead{Martine.Leonard@univ-rouen.fr}

\author[Rouen,Poly]{Laurent Mouchard}

    \ead{Laurent.Mouchard@univ-rouen.fr}

    \address[Poly]{Laboratoire d'Informatique de l'Ecole Polytechnique (LIX),
        CNRS UMR 7161, France}

\author[SNU]{Kunsoo Park\corref{cor1}}

    \ead{kpark@theory.snu.ac.kr}

    \cortext[cor1]{Corresponding author.}

\newpage



\begin{abstract}
Recently, a compressed index for similar strings, called the {\em FM-index of alignment} (FMA),
 has been proposed with the functionalities of pattern search and random access.
The FMA is quite efficient in space requirement and pattern search time,
 but it is applicable only for an alignment of similar strings without gaps.
In this paper we propose the {\em FM-index of alignment with gaps},
 a realistic index for similar strings, which allows gaps in their alignment.
For this,
 we design a new version of the suffix array of alignment
 by using alignment transformation and a new definition of the alignment-suffix.
The new suffix array of alignment
 enables us to support the LF-mapping and backward search, the key functionalities of the FM-index,
 regardless of gap existence in the alignment.
We experimentally compared our index with RLCSA due to M\"{a}kinen et al.
 on 100 genome sequences from the 1000 Genomes Project.
The index size of our index is less than one third of that of RLCSA.
\end{abstract}


\end{frontmatter}

\pagenumbering{arabic} \setcounter{page}{1}


\section{Introduction}


\partitle{Problem def. and Previous works}

A lot of indexes not only storing similar strings but also supporting efficient pattern search
 have been developed
 such as RLCSA~\cite{recomb/MakinenNSV09,jcb/MakinenNSV10,spire/SirenVMN08},
 LZ-scheme based indexes~\cite{tcs/DoJSS14,rsta/FerradaGHP14,tcs/KreftN13}
 compressed suffix trees~\cite{spire/AbeliukN12,sea/NavarroP14},
 and so on~\cite{aaim/HuangLSTY10,iwoca/Navarro12}.
To exploit the similarity of the given strings,
 most of them use classical compression schemes
 such as run-length encoding and Lempel-Ziv compressions~\cite{spire/KuruppuPZ10,Ziv&Lempel:77}.
Recently,
 Na et al.~\cite{tcs/NaKPLLMP16,iwoca/NaPCHIMP13,spire/NaPLHLMP13}
 took a new approach using an alignment of similar strings without classical compression schemes,
 and they proposed indexes of alignment called
 the {\em suffix tree of alignment}~\cite{iwoca/NaPCHIMP13},
 the {\em suffix array of alignment} (SAA)~\cite{spire/NaPLHLMP13}, and
 the {\em FM-index of alignment} (FMA)~\cite{tcs/NaKPLLMP16}.
The FMA, a compressed version of the SAA, is the most efficient among the three indexes
 but it is applicable only for an alignment of similar strings {\em without gaps}.

\partitle{motivation - vcf}

However, real-world data include gaps in alignments.
Figure~\ref{fig:vcf} shows Variant Call Format (VCF) files
 created by SAMtools (Sequence Alignment/Map tools)
 for sequences from the 1000 Genomes Project~\cite{nature/1000Genomes10}.
A VCF file contains alignment information
 between an individual sequence and its reference sequence.
Note that not only substitutions but also indels (insertions and deletions)
 are contained in an alignment.
For example, the first line of the `VCF 3' file in Fig.~\ref{fig:vcf}
 says that {\tt AT} at position 786763 in the reference sequence
 is aligned with {\tt A} in the individual sequence.
Thus, the FMA~\cite{tcs/NaKPLLMP16} allowing only substitutions in an alignment
 is an unrealistic index.

\begin{figure}
\centering
\includegraphics[scale=0.5]{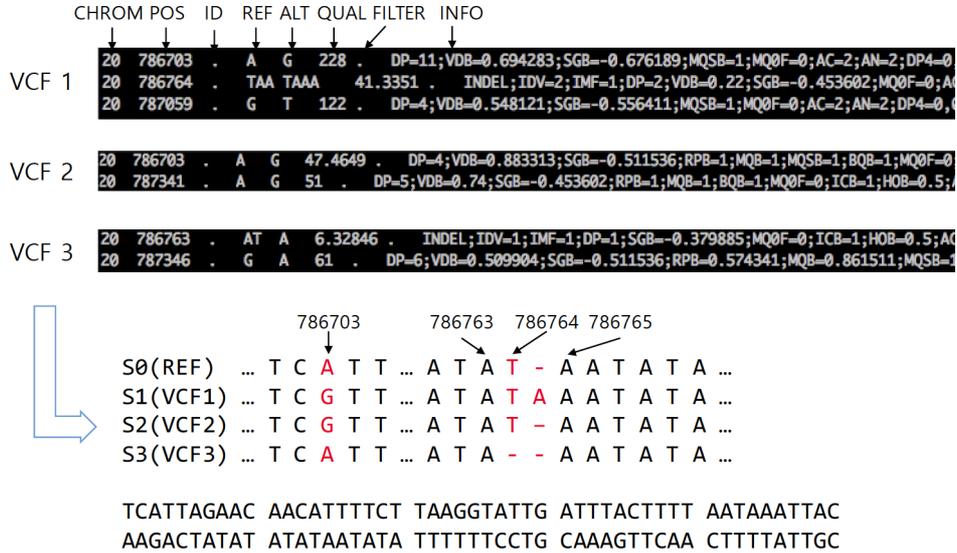}
\caption{Example of a VCF file.}
\label{fig:vcf}
\end{figure}

\partitle{our contribution}

In this paper we propose a new FM-index of alignment,
 a realistic compressed index for similar strings,
 allowing indels as well as substitutions in an alignment.
(We call our index the FMA with gaps and the previous version the FMA without gaps.)
For this, we design a new version of the SAA by using an alignment transformation and
 a new definition of the suffix of an alignment (called the {\em alignment-suffix}).
In our index,
 an alignment is divided into two kinds of regions, common regions and non-common regions,
 and gaps in a non-common region are put together into one gap in the transformed alignment.
The alignment-suffix is defined for the transformed alignment
 but its definition is different from those defined
 in~\cite{tcs/NaKPLLMP16,iwoca/NaPCHIMP13,spire/NaPLHLMP13}.
Due to the alignment transformation and the new definition of the alignment-suffix,
 our index supports the LF-mapping and backward search,
 the key functionalities of the
 FM-index~\cite{focs/FerraginaM00,soda/FerraginaM01,jacm/FerraginaM05},
 regardless of gap existence in the alignment.

\partitle{common and NC}

For constructing our index, we must find common regions and non-common regions for the given strings
 but we do not need to find a multiple alignment for the given strings
 since the knowledge about positions where substitutions and indels occur
 are of no use in our transformed alignment.
Finding common and non-common regions is much easier and simpler than finding a multiple alignment.
For instance, common regions and non-common regions
 between an individual sequence and its reference sequence
 can be directly obtained from a VCF file.
In the example of Fig.~\ref{fig:vcf}, position 786703 is a non-common region
 and positions 786704..786763 are a common region.
Hence,  based on the reference sequence,
 common regions and non-common regions of 100 genome sequences
 can be easily created.

We implemented the FMA with gaps and did experiments
 on 100 genome sequences from the 1000 Genomes Project.
We compared our FMA with RLCSA
 due to M\"{a}kinen et al.~\cite{jcb/MakinenNSV10}.
The index size of our FMA is less than one third of that of RLCSA.
Our index is faster in pattern search and
 RLCSA is faster in random access.

\partitle{Organization}

This paper is organized as follows.
We first describe our FMA and search algorithm
 for an alignment with gaps in Section 2
 and give experimental results in Section 3.
In Section 4, we conclude with remarks.


\section{FM-index of alignment with gaps}
\label{sec:FMA}



\subsection{Alignments with gaps}
\label{subsec:alignment}


\partitle{intro.}

Consider a multiple alignment in Fig.~\ref{fig:alignment} (a)
 of four similar strings:
    $S^1 = $ {\tt \$cct\C\A aac\C \#},
    $S^2 = $ {\tt \$cct\C\C\A aac\A \#},
    $S^3 = $ {\tt \$cct\T\A\T aac\underline{}\#}, and
    $S^4 = $ {\tt \$cct\underline{}aac\C \#}.
These strings are the same except the underlined characters
 and one string can be transformed into another strings
 by replacing, inserting or deleting underlined substrings.
Formally, we are given an alignment $\Upsilon$ of $m$ similar strings
 $S^j = \alpha_1 \Delta^j_1 \ldots \alpha_r \Delta^j_r \alpha_{r+1}$ ($1 \le j \le m$)
  over an alphabet $\Sigma$,
  where $\alpha_i$ $(1\le i \le r+1)$ is a common substring in all strings and
 $\Delta^j_i$ ($1\le i \le r$) is a non-common substring in string $S^j$.
In the example above,  $\alpha_1 =$ {\tt \$cct}, $\alpha_2 =$ {\tt aac}, $\alpha_3 =$ {\tt \#}.
Without loss of generality,
 we assume $\alpha_1$ starts with {\tt \$} and $\alpha_{r+1}$ ends with {\tt \#}
 where {\tt \$} and {\tt \#} are special symbols occurring nowhere else in $S^j$,
 and each $\alpha_i$ is not empty.
%

%

\begin{figure}[t]
\renewcommand{\arraystretch}{0.9}
\setlength{\tabcolsep}{0.8mm}
\setlength{\doublerulesep}{0.2pt}
%
{\centering \tt 
\qquad
\begin{tabular}[t]{ r cc cc  cccc | c cc |c| c }
 \small \rm pos.
  & \footnotesize $1$ & \footnotesize $2$ & \footnotesize $3$ & \footnotesize $4$
  & \multicolumn{4}{|c|}{\footnotesize\rm 5~ 6~ 7~ 8}
  & \footnotesize $9$ & \footnotesize $0$ & \footnotesize $1$ & \footnotesize $2$ & \footnotesize $3$ \\
%
 $S^1 =$ & \$ & c & c & t & \multicolumn{4}{|c|}{\underline{c - a -}} & a & a & c & \C & \# \\
 $S^2 =$ & \$ & c & c & t & \multicolumn{4}{|c|}{\underline{c c a -}} & a & a & c & \A & \# \\
 $S^3 =$ & \$ & c & c & t & \multicolumn{4}{|c|}{\underline{t - a t}} & a & a & c & \underline{-} & \# \\
 $S^4 =$ & \$ & c & c & t & \multicolumn{4}{|c|}{\underline{- - - -}} & a & a & c & \C & \# \\
%
 & \multicolumn{2}{c}{\small $\tilde{\alpha}^{\diamond}_1$}
 & \multicolumn{2}{|c}{\small$\tilde{\alpha}^{+}_1$} & \multicolumn{4}{|c|}{\small$\Delta_1$}
 & \multicolumn{1}{c|}{\small$\tilde{\alpha}^{\diamond}_2$}
 & \multicolumn{2}{c|}{\small$\tilde{\alpha}^{+}_2$} & \multicolumn{1}{c|}{\small$\Delta_2$}
 & \multicolumn{1}{c}{\small$\alpha_3$} \\
%
%
 \multicolumn{14}{c}{\rm (a)}
\end{tabular}
\qquad
\begin{tabular}[t]{ r cc cc ccc| c| cc c| c }
 \small \rm pos.
  & \footnotesize $1$ & \footnotesize $2$
  & \multicolumn{5}{|c|}{\footnotesize\rm ~3 ~4 ~5 ~6 ~7 ~}
  & \footnotesize $8$
  & \footnotesize $9$ & \footnotesize $0$ & \footnotesize $1$ & \footnotesize $2$ \\
%
 $S^1 =$ & \$ & c & \multicolumn{5}{|c|}{- c t \underline{c a}} & a & a & c & \C & \# \\
 $S^2 =$ & \$ & c & \multicolumn{5}{|c|}{c t \underline{c c a}} & a & a & c & \A & \# \\
 $S^3 =$ & \$ & c & \multicolumn{5}{|c|}{c t \underline{t a t}} & a & - & a & c & \# \\
 $S^4 =$ & \$ & c & \multicolumn{5}{|c|}{- - - c t} & a & a & c & \C & \# \\
%
 & \multicolumn{2}{c}{\small$\tilde{\alpha}^{\diamond}_1$}
 & \multicolumn{5}{|c|}{\small$\tilde{\alpha}^{+}_1 \Delta_1$}
 & \multicolumn{1}{c|}{\small$\tilde{\alpha}^{\diamond}_2$}
 & \multicolumn{3}{c|}{\small$\tilde{\alpha}^{+}_2 \Delta_2$}
 & \multicolumn{1}{c}{\small$\alpha_3$} \\
%
%
 \multicolumn{13}{c}{\rm (b)}
\end{tabular}
}

\caption{\label{fig:alignment}
 An example of (a) an original alignment and (b) its transformed alignment.
}
\end{figure}


\partitle{alpha+}

For each common substring $\alpha_i$,
we define $\tilde{\alpha}^{+}_{i}$ as follows\footnote{Note that
 the definition of $\tilde{\alpha}^{+}_{i}$ is different from that of $\tilde{\alpha}^{*}_{i}$
 in~\cite{tcs/NaKPLLMP16,iwoca/NaPCHIMP13,spire/NaPLHLMP13}.
The $\tilde{\alpha}^{+}_{i}$ is longer than $\tilde{\alpha}^{*}_{i}$ by one.}.
\begin{definition}
The string $\tilde{\alpha}^{+}_{i}$ $(1 \le i \le r)$ is the shortest suffix of $\alpha_i$
occurring only once in each string $S^{j}$ ($1\le j \le m$)
and $\tilde{\alpha}^{+}_{r+1}$ is an empty string.
\end{definition}
Consider $\alpha_{1}= \mbox{\tt \$cct}$ in Fig.~\ref{fig:alignment}.
Since the suffix {\tt t} of length 1 occurs more than once in $S^3$
 but the suffix {\tt ct} of length 2 occurs only once in each string,
 $\tilde{\alpha}^{+}_{1}$ is {\tt ct}.
Similarly, $\tilde{\alpha}^{+}_{2}$ is {\tt ac},
 which is the shortest suffix of $\alpha_2$ occurring only once in each string.
Without loss of generality, for $2\le i \le r+1$,
 $\tilde{\alpha}^{+}_{i}$ is assumed to be shorter than $\alpha_{i}$.
(If $\tilde{\alpha}^{+}_{i}$ is equal to $\alpha_{i}$, we merge $\alpha_i$
 with its adjacent non-common substrings $\Delta^j_{i-1}$ and $\Delta^j_{i}$,
 and regard $\Delta^j_{i-1} \alpha_{i} \Delta^j_i$ as one non-common substring).

\partitle{Alignment transformation}

For indexing similar strings whose alignment includes gaps,
 we first transform the given alignment $\Upsilon$ into its right-justified form $\widetilde{\Upsilon}$
 so that the characters
in each $\tilde{\alpha}^+_i \Delta^j_i$ ($ 1\le i \le r$, $ 1 \le j \le m$)
 are right-justified.
See Fig.~\ref{fig:alignment} for an example,
 where a gap is represented by a series of hyphens `{\tt -}'
(note that `{\tt -}' is not a character).
Hereafter, to indicate positions of characters in $S^j$,
 we use the positions in the transformed alignment $\widetilde{\Upsilon}$
 and denote by $\widetilde{S}^j[i]$ the character of $S^j$ at the $i$th position in $\widetilde{\Upsilon}$.
If $\widetilde{S}^j[i]$ is `{\tt -}',
 we say $\widetilde{S}^j[i]$ is empty.
The positions in ${S}^j$ and $\widetilde{S}^j$ can be easily converted into each other
 by storing gap information.
Moreover, we denote the suffix of $\widetilde{S}^j$ starting at position $q$ by suffix $(j, q)$,
e.g., the suffix $(3, 8)$ is {\tt aac\#} in Fig.~\ref{fig:alignment}.

\partitle{rep. of alignment}

An alignment of similar strings can be compactly represented by combining each common substring $\alpha_i$ in all strings
 as in~\cite{tcs/NaKPLLMP16,iwoca/NaPCHIMP13,spire/NaPLHLMP13}.
However, the representation is not suitable for the transformed alignment $\widetilde{\Upsilon}$
because the characters in $\tilde{\alpha}^+_i$ are not aligned in $\widetilde{\Upsilon}$.
Thus, we introduce another representation.
Let $\tilde{\alpha}^{\diamond}_i$ $(1 \le i \le r+1)$ be the prefix of $\alpha_i$
 such that $\alpha_i = \tilde{\alpha}^{\diamond}_i \tilde{\alpha}^+_i$.
Then, we represent the transformed alignment $\widetilde{\Upsilon}$
  by combining $\tilde{\alpha}^{\diamond}_i$ (rather than $\alpha_i$):
 $\widetilde{\Upsilon} = \tilde{\alpha}^{\diamond}_{1}
  (\tilde{\alpha}^{+}_{1} \Delta^1_1 / \cdots / \tilde{\alpha}^{+}_{1} \Delta^m_1)
   \cdots  \tilde{\alpha}^{\diamond}_{r}
  (\tilde{\alpha}^{+}_{r} \Delta^1_r / \cdots /  \tilde{\alpha}^{+}_{r} \Delta^m_r)
  \tilde{\alpha}^{\diamond}_{r+1}$.
The alignment in Fig.~\ref{fig:alignment} is represented as
 $\widetilde{\Upsilon} =$ {\tt \$c(ct\C\A /ct\C\C\A/ct\T\A\T/ct)a(ac\C/ac\A/ac/ac\C)\#}.
We denote
  $(\tilde{\alpha}^{+}_{i} \Delta^1_i / \cdots / \tilde{\alpha}^{+}_{i} \Delta^m_i)$
  by $\tilde{\alpha}^{+}_{i} \Delta_i$
 and call it a {\em ps-region} (partially-shared region).
Also, we call $\tilde{\alpha}^{\diamond}_{i}$ a {\em cs-region} (completely-shared region).


\subsection{Suffix array and FM-index of alignment with gaps}
\label{subsec:FMA}


\partitle{a-suffixes}

In our index, one or more suffixes starting at an identical position $q$
 are compactly represented by one {\em alignment-suffix} (for short {\em a-suffix})
 defined as follows.
We have two cases according to whether the starting position $q$ is
 in a cs-region or a ps-region.
\bi
\item The case when $q$ is in a cs-region $\tilde{\alpha}^{\diamond}_i$ ($1\le i \le r+1$).
    Let $\alpha'_i$ be the suffix of $\tilde{\alpha}^{\diamond}_i$ starting at $q$.
    All the suffixes starting at $q$ is represented by the a-suffix
     $\alpha'_i(\tilde{\alpha}^{+}_{i} \Delta^{1}_i / \cdots / \tilde{\alpha}^{+}_{i} \Delta^{m}_i )
      \cdots$.
    In the previous example, the suffixes starting at position 8
     are represented by the a-suffix {\tt a(ac\C/ac\A/ac/ac\C)\#}.

\item The case when $q$ is in a ps-region $\tilde{\alpha}^+_i \Delta_i$ ($1\le i \le r$).
    Let $\delta^{j}_i$ ($1\le j \le m$) be the suffix of $\tilde{\alpha}^+_i \Delta^j_i$ starting at $q$.
    Then, the set of the suffixes starting at $q$ is partitioned
     so that the suffixes of $S^{j_1}$ and $S^{j_2}$ are in the same subset
     if and only if $\delta^{j_1}_i =  \delta^{j_2}_i$.
    For each subset
     $\{ \delta^{j_1}_i \cdots \tilde{\alpha}^{\diamond}_{r+1}, \ldots ,
      \delta^{j_k}_i\cdots \tilde{\alpha}^{\diamond}_{r+1} \}$,
     all the suffixes in the subset are represented by the a-suffix
     $(\delta^{j_1}_i / \cdots / \delta^{j_k}_i )\cdots \tilde{\alpha}^{\diamond}_{r+1}$.
    For example, the set of the suffixes starting at position 9 is partitioned into
     two subsets $\{\widetilde{S}^1[9..12], \widetilde{S}^4[9..12] \}$ and $\{\widetilde{S}^2[9..12] \}$,
    and they are represented by the a-suffixes
     {\tt (ac\C/ac\C)\#} and {\tt ac\A \#}, respectively.
    Note that no suffix of $\widetilde{S}^3$ starts at position 9.

\ei
The suffixes represented by an a-suffix appear consecutively
 in the generalized suffix array of the given strings
 since $\tilde{\alpha}^{+}_i$ occurs only once in each given string.
Note that a suffix of $\tilde{\alpha}^{+}_i$ may occur more than once in a string,
 and thus $\tilde{\alpha}^{+}_{i}$ does not belong to a cs-region but to a ps-region.

\begin{figure}[t]
\setlength{\tabcolsep}{1.3mm}
{\centering \small
\begin{tabular}{c|c c|c|l|c| c c | c c | c c}
\hline 
$idx$ & \multicolumn{2}{c|}{SAA} & $F$ & \multicolumn{1}{c|}{a-suffixes} & $L$ & \multicolumn{6}{c}{$occ(\sigma,i)$ \& $B_{\sigma}$} \\
      & $strs$ & $pos$ &     &\multicolumn{1}{c|}{(cyclic shifts)}&  & \multicolumn{2}{c}{{\tt a}}&\multicolumn{2}{c}{{\tt c}}&\multicolumn{2}{c}{{\tt t}} \\
\hline 
1&  0&  1&  {\tt \$}&   {\tt  \$c(ct\C\A /ct\C\C\A/ct\T\A\T/ct)a(ac\C/ac\A/ac/ac\C)\#}&     {\tt \#}&   0&  &   0&  &   0&   \\
2&  0&  12& {\tt \#}&   {\tt \#\$c(ct\C\A /ct\C\C\A/ct\T\A\T/ct)a(ac\C/ac\A/ac/ac\C) }&     {\tt a,c}&  1&  &   1&  &   0&   \\
3&  2&  11& {\tt a }&   {\tt \A \#\$cct\C \C \A aac}&                                       {\tt c}&    1&  &   2&  &   0&   \\
4&  1,2&    7&  {\tt a }&   {\tt (\A/\A)a(ac\C/ac\A)\#\$c(ct\C/ct\C\C) }&                   {\tt c}&    1&  &   3&  &   0&   \\
5&  0&  8&  {\tt a}&    {\tt a(ac\C/ac\A/ac/ac\C)\#\$c(ct\C\A /ct\C\C\A/ct\T\A\T/ct)}&      {\tt a,t}&  2&  &   3&  &   1&   \\
6&  3&  10& {\tt a}&    {\tt ac\#\$cct\T \A \T a}&                                          {\tt a}&    3&  1&  3&  &   1&   \\
7&  2&  9&  {\tt a}&    {\tt ac\A \#\$cct\C \C \A a}&                     {\tt $\langle$a$\rangle$}&    3&  1&  3&  &   1&   \\
8&  1,4&    9&  {\tt a}&    {\tt (ac\C/ac\C)\#\$c(ct\C\A/ct)a}&           {\tt $\langle$a$\rangle$}&    3&  1&  3&  &   1&   \\
9&  3&  6&  {\tt a }&   {\tt \A \T aac\#\$cct\T }&                                          {\tt t}&    3&  &   3&  &   2&   \\
10& 1,3,4&  11& {\tt c }&   {\tt (\C/c/\C)\#\$c(ct\C\A/ct\T\A\T/ct)a(ac/a/ac)}&             {\tt a,c}&  4&  &   4&  &   2&   \\
11& 2&  10& {\tt c}&    {\tt c\A \#\$cct\C \C \A aa}&                                       {\tt a}&    5&  &   4&  &   2&   \\
12& 1,2&    6&  {\tt c }&   {\tt (\C\A/\C\A)a(ac\C/ac\A)\#\$c(ct/ct\C) }&                   {\tt c,t}&  5&  &   5&  &   3&   \\
13& 1,4&    10& {\tt c}&    {\tt (c\C/c\C)\#\$c(ct\C\A/ct)a(a/a)}&                          {\tt a}&    6&  &   5&  &   3&   \\
14& 2&  5&  {\tt c}&    {\tt \C \C \A aac\A \#\$cct}&                                       {\tt t}&    6&  &   5&  &   4&   \\
15& 0&  2&  {\tt c}&    {\tt c(ct\C\A /ct\C\C\A/ct\T\A\T/ct)a(ac\C/ac\A/ac/ac\C)\#\$}&      {\tt \$}&   6&  &   5&  &   4&   \\
16& 4&  6&  {\tt c}&    {\tt ctaac\C \#\$c}&                                                {\tt c}&    6&  &   6&  1&  4&   \\
17& 1&  4&  {\tt c}&    {\tt ct\C \A aac\C \#\$c}&                        {\tt $\langle$c$\rangle$}&    6&  &   6&  1&  4&   \\
18& 2&  3&  {\tt c}&    {\tt ct\C \C \A aac\A \#\$c}&                     {\tt $\langle$c$\rangle$}&    6&  &   6&  1&  4&   \\
19& 3&  3&  {\tt c}&    {\tt ct\T \A \T aac\#\$c}&                        {\tt $\langle$c$\rangle$}&    6&  &   6&  1&  4&   \\
20& 3,4&    7&  {\tt t }&   {\tt (\T/t)a(ac/ac\C)\#\$c(ct\T\A/c)}&                          {\tt a,c}&  7&  &   7&  &   4&   \\
21& 3&  5&  {\tt t }&   {\tt \T \A \T aac\#\$cct}&                                          {\tt t}&    7&  &   7&  &   5&   \\
22& 1&  5&  {\tt t}&    {\tt t\C \A aac\C \#\$cc}&                                          {\tt c}&    7&  &   8&  &   5&   \\
23& 2&  4&  {\tt t}&    {\tt t\C \C \A aac\A \#\$cc}&                                       {\tt c}&    7&  &   9&  &   5&   \\
24& 3&  4&  {\tt t}&    {\tt t\T \A \T aac\#\$cc}&                                          {\tt c}&    7&  &   10& &   5&   \\
\hline 
\end{tabular}
}
\caption{\label{fig:FMA}
The SAA and FMA for
$\widetilde{\Upsilon} =$ {\tt \$c(ct\C\A /ct\C\C\A/ct\T\A\T/ct)a(ac\C/ac\A/ac/ac\C)\#}.
 (Bit 0 is omitted in $B_{\sigma}$.)
}
\end{figure}


\partitle{SAA \& BWT}

The {\em suffix array of alignment (SAA)}
 is a lexicographically sorted array of all the a-suffixes of
 the transformed alignment $\widetilde{\Upsilon}$.
See Fig.~\ref{fig:FMA} for an example,
 where the string number 0 indicates the string numbers $1, \ldots, m$.
We denote by $SAA[i]$ the $i$th entry of the SAA.
Let us consider a-suffixes in the SAA as cyclic shifts (rotated alignments)
 as in the Burrows-Wheeler transform~\cite{Burrows&Wheeler:94}.
Then, the array $F[i]$ (resp. $L[i]$)
 is the set of the first (resp. last) characters of the suffixes
  represented by the a-suffix in $SAA[i]$.
By definition of the a-suffixes,
 the first characters of the suffixes represented by an a-suffix are of the same value
 and thus $F[i]$ has one element.
However, $L[i]$ may have more than one element (at most $|\Sigma|$ elements)
 as shown in Fig.~\ref{fig:FMA}.
For example, $F[13] = \{ \widetilde{S}^1[10], \widetilde{S}^4[10] \} = \{\mbox{\tt c}\}$
 and $L[10] = \{ \widetilde{S}^1[10], \widetilde{S}^3[10], \widetilde{S}^4[10] \}
   = \{ \mbox{\tt a,c} \}$.
Since gaps are not considered as characters in $\widetilde{\Upsilon}$,
 when letting $q$ be the position of the characters in $F[i]$,
 the positions of the characters in $L[i]$ may be less than $q-1$.
(On the other hand, the position of the characters in $L[i]$
 is always $q-1$ for an alignment without gaps when $q>1$.)
In Fig.~\ref{fig:FMA}, $F[17] = \widetilde{S}^{1}[4]$ and
  $L[17] = \widetilde{S}^{1}[2]$ because $\widetilde{S}^{1}[3]$ is empty.

\partitle{LF-mapping}

We define the LF-mapping for the arrays $L$ and $F$.
Let $\mathcal{L}$ be the set of pairs
 of a character $\sigma$ and an entry index $i$ such that $\sigma \in L[i]$.
In the example of Fig.~\ref{fig:FMA},
 $\mathcal{L} = \{ (\mbox{\tt \#},1), (\mbox{\tt a},2), (\mbox{\tt c},2), (\mbox{\tt c},3), (\mbox{\tt c},4),
 (\mbox{\tt a},5),  (\mbox{\tt t},5), \ldots \} $.
For a pair $(\sigma, i) \in \mathcal{L}$,
 the {\em LF-mapping $LF(\sigma, i)$}
 is defined as the index of $F[k]$ containing the characters corresponding to $\sigma$ in $L[i]$.
For example, see $L[10] =\{ \widetilde{S}^1[10], \widetilde{S}^3[10], \widetilde{S}^4[10] \} = \{ \mbox{\tt a},\mbox{\tt c} \}$.
Since {\tt a} in $L[10]$ corresponds to $\widetilde{S}^{3}[10]$
 and it is contained in $F[6]$,
 $LF(\mbox{\tt a}, 10) = 6$.
Similarly, $LF(\mbox{\tt c}, 10) = 13$
 since $\widetilde{S}^{1}[10]$ and $\widetilde{S}^{4}[10]$ (i.e., {\tt c} in $L[10]$)
 are contained in $F[13]$.
Note that  $LF(\mbox{\tt c}, 10)$ is well defined
 since the characters in $L[10]$ whose values are {\tt c}
 are all contained in an {\em identical} entry $F[13]$.
This is always true in the transformed alignment $\widetilde{\Upsilon}$
 even though gaps exist in $\widetilde{\Upsilon}$,
 as shown in the following lemma.
(It is not true in the untransformed alignment $\Upsilon$.)
\begin{lemma} \label{lem:LF}
For a pair $(\sigma,i) \in \mathcal{L}$,
 the characters in $L[i]$ whose values are $\sigma$
 are all contained in an identical entry of $F$.
\end{lemma}
\begin{proof}
Let $q$ be the starting position of the suffixes in $SAA[i]$,
 and $\widetilde{S}^{j_1}[q_1]$ and $\widetilde{S}^{j_2}[q_2]$ ($j_1 \neq j_2$) be
 two characters in $L[i]$ whose values are $\sigma$.
Without loss of generality, we assume $q>1$.
Then, $q_1$ and $q_2$ are less than $q$.
We have three cases according to
 whether $\widetilde{S}^{j_1}[q-1]$ and $\widetilde{S}^{j_2}[q-1]$ are empty.
\bi
\item
    First, when none of $\widetilde{S}^{j_1}[q-1]$ and $\widetilde{S}^{j_2}[q-1]$ are empty
    (i.e., $q_1 = q_2 = q-1$),
    $\widetilde{S}^{j_1}[q_1]$ and $\widetilde{S}^{j_2}[q_2]$ are contained in an identical entry of $F$
    by definition of the a-suffix,
    which can be shown as in~\cite{tcs/NaKPLLMP16}.
\item
    Second, when both of $\widetilde{S}^{j_1}[q-1]$ and $\widetilde{S}^{j_2}[q-1]$ are empty,
     both $\widetilde{S}^{j_1}[q_1]$ and $\widetilde{S}^{j_2}[q_2]$
     are the last character in a cs-region $\tilde{\alpha}^{\diamond}_i$
     since the characters in ps-region $\tilde{\alpha}^{+}_i \Delta_i$ are right-justified in $\widetilde{\Upsilon}$.
    Thus, $q_1 = q_2$ and by definition of the a-suffix,
     the suffixes $(j_1,q_1)$ and $(j_2,q_2)$ are contained in an identical entry of the SAA.
    Hence, $\widetilde{S}^{j_1}[q_1]$ and $\widetilde{S}^{j_2}[q_2]$
     are contained in an identical entry of $F$.
\item
    The third case is when only one of $\widetilde{S}^{j_1}[q-1]$ and $\widetilde{S}^{j_2}[q-1]$ is empty.
    We show by contradiction that this case cannot happen.
    Without loss of generality,
     assume $\widetilde{S}^{j_1}[q-1]$ is empty and $\widetilde{S}^{j_2}[q-1]$ is not empty.
    Since $\widetilde{S}^{j_1}[q-1]$ is empty,
     $\widetilde{S}^{j_1}[q_1]$ is the last character in a cs-region $\tilde{\alpha}^{\diamond}_k$
     and $\widetilde{S}^{j_1}[q]$ is the first character in ps-region $\tilde{\alpha}^{+}_k \Delta_k$.
    It means that the suffix $(j_1,q)$ is prefixed by $\tilde{\alpha}^{+}_k$.
    Since both suffixes $(j_1,q)$ and $(j_2,q)$ are in $SAA[i]$,
     by definition of the a-suffix,
     the suffix $(j_2,q)$ is also prefixed by $\tilde{\alpha}^{+}_k$.
    Since $\tilde{\alpha}^{+}_k$ occurs only once in each string,
      $\widetilde{S}^{j_2}[q_2]$ is the last character in $\tilde{\alpha}^{\diamond}_k$
     (i.e., $q_1 = q_2$) and
     $\widetilde{S}^{j_2}[q_2+1..q-1]$ is empty.
    It contradicts with the assumption that $\widetilde{S}^{j_2}[q-1]$ is not empty.
\ei
Therefore, the characters in $L[i]$ whose values are $\sigma$
 are all contained in an identical entry of $F$.
\end{proof}
For a character $\sigma \in \Sigma$,
 a pair $(\sigma, i) \in \mathcal{L}$ will be called an $\mathcal{L}_\sigma$-pair.
For two $\mathcal{L}_\sigma$-pairs $(\sigma,i)$ and $(\sigma,i')$,
 we say that $(\sigma, i)$ is smaller than $(\sigma, i')$
 if and only if $i < i'$.

\partitle{two types of LF-mapping}

The LF-mapping $LF(\sigma,i)$ is not a one-to-one correspondence.
Multiple pairs can be mapped to the same entry in $F$.
See $L[6] = \{\widetilde{S}^3[8]\} = \{\mbox{\tt a} \}$,
 $L[7] = \{\widetilde{S}^2[8]\} = \{\mbox{\tt a} \}$,
 and $L[8] = \{\widetilde{S}^1[8], S^4[8]\} = \{\mbox{\tt a} \}$.
Since all of them are {\tt a} in $F[5]$,
 $LF(\mbox{\tt a}, 6) = LF(\mbox{\tt a}, 7) = LF(\mbox{\tt a}, 8) = 5$.
Thus, we classify pairs $(\sigma,i) \in \mathcal{L}$ into two types:
A pair $(\sigma,i) \in \mathcal{L}$ is an {\em (m:1)-type} (many-to-one mapping-type)
pair if there exists another pair $(\sigma,i') \in \mathcal{L}$
 such that $LF(\sigma,i) = LF(\sigma,i')$;
 otherwise, $(\sigma,i)$ is a {\em (1:1)-type} (one-to-one mapping-type) pair.
The following lemma shows that for a (m:1)-type pair $(\sigma,i)$,
 the last characters of all the suffixes in $SAA[i]$ are mapped to an identical entry in $F$.
(This lemma is necessary for our search algorithm to work correctly
 and it is also satisfied for the FMA without gaps~\cite{tcs/NaKPLLMP16}.
 However, it is not satisfied when defining our index
 using $\tilde{\alpha}^{*}$ in~\cite{tcs/NaKPLLMP16} rather than
 $\tilde{\alpha}^{+}$.)
\begin{lemma} \label{lem:m-1}
If a pair $(\sigma,i) \in \mathcal{L}$ is of (m:1)-type,
 the values of the characters in $L[i]$ are all the same $\sigma$.
\end{lemma}
\begin{proof}
Let $k=LF(\sigma, i)$ and $q_k$ be the starting position of the suffixes in $SAA[k]$.
By definition of the a-suffix,
 the pair $(\sigma,i)$ is of (m:1)-type
 only if $q_k$ is the last position in a cs-region $\tilde{\alpha}^{\diamond}_j$
 and the last character in $\tilde{\alpha}^{\diamond}_j$ is $\sigma$.
Hence, all the suffixes in $SAA[i]$ are prefixed by $\tilde{\alpha}^{+}_j$.
Since $\tilde{\alpha}^{+}_j$ occurs only once in each string,
 the preceding character of $\tilde{\alpha}^{+}_j$
 is the last character in $\tilde{\alpha}^{\diamond}_j$, i.e., $\sigma$.
Therefore, $L[i]$ has only one character $\sigma$.
\end{proof}
To handle (m:1)-type pairs in $\mathcal{L}$ efficiently,
 we define bit-vectors $B_{\sigma}$'s as follows:
 $B_{\sigma}[i] = 1$
 if and only if $(\sigma,i)$ is in $\mathcal{L}$
 and it is of (m:1)-type
 (see Fig.~\ref{fig:FMA}).

\partitle{LF-mapping computation}

The LF-mapping can be easily computed
 using the array $C$ and the function $occ$ defined
 as follows~\cite{tcs/NaKPLLMP16}.
\bi
\item For $\sigma \in \Sigma$,
     $C[\sigma]$ is the total number of entries in $F$
      containing characters alphabetically smaller than $\sigma$.
     $C[|\Sigma|+1]$ is the size of $F$.
\item For a character $\sigma\in \Sigma$ and an entry index $i$ in the SAA,
    $occ(\sigma, i)$ is the number of $\mathcal{L}_{\sigma}$-pairs $(\sigma, i')$
    such that $i' \le i$, i.e.,
    the number of entries in $L[1..i]$ containing the character $\sigma$.
    If more than one pair $(\sigma, i') \in \mathcal{L}$
     are mapped to an identical entry in $F$,
     we count only the smallest $\mathcal{L}_{\sigma}$-pair among them.
    For example, consider $occ(\mbox{\tt a},i)$ for $i=6, 7, 8$.
    Since {\tt a}'s in $L[6..8]$ are all contained in $F[5]$,
     we only count $(\mbox{\tt a},6)$
     and thus $occ(\mbox{\tt a},i)$'s are the same for $i = 6, 7, 8$.
    In Fig.~\ref{fig:FMA}, uncounted characters in $L$
     are indicated by $\langle ~ \rangle$.
\ei
Then, $LF(\sigma,i) = C[\sigma] + occ(\sigma,i)$.
See $L[10]$ in Fig.~\ref{fig:FMA},
 which has two $\mathcal{L}_\sigma$-pairs $(\mbox{\tt a},10)$ and $(\mbox{\tt c},10)$.
We have $LF(\mbox{\tt a},10) = C[\mbox{\tt a}] + occ(\mbox{\tt a},10) = 2 + 4 = 6$
 and $LF(\mbox{\tt c},10) = C[\mbox{\tt c}] + occ(\mbox{\tt c},10) = 9 + 4 = 13$.
%


\subsection{Pattern Search}
\label{subsec:Search}


\partitle{intro. and range ($\First_\ell$, $\Last_\ell$)}

Pattern search is to find all occurrences of a given pattern $P[1..p]$
 in the given strings $S^1,\ldots, S^m$.
Our pattern search algorithm
 proceeds backward using the LF-mapping with the array $C$ and the function $occ$.
It consists of at most $p$ steps from Step $p$ to Step $1$.
In Step $\ell = p, \ldots, 1$,
 the algorithm finds the range $(\First_\ell, \Last_\ell)$ in the SAA
 defined as follows:
\be
\item[i)] $\First_p$ (resp. $\Last_p$) is the smallest (resp. largest) index $i$
  such that $F[i] = \{P[p]\}$.
\item[ii)] For $\ell=p-1,\ldots, 1$,
  $\First_\ell$ (resp. $\Last_\ell$) is the LF-mapping value of
  the smallest (resp. largest) $\mathcal{L}_{\sigma}$-pair
  in the range $(\First_{\ell+1}, \Last_{\ell+1})$,
 where $\sigma = P[\ell]$.
 If there exists no $\mathcal{L}_{\sigma}$-pair in $(\First_{\ell+1}, \Last_{\ell+1})$,
  then we set $\First_\ell = \Last_\ell+1$.
\ee
Then, all the suffixes prefixed by $P[\ell..p]$
 are in $SAA[\First_\ell..\Last_\ell]$ and
 the size of the range decreases monotonically
 when $\ell$ decreases.

\partitle{set $\Z_\ell$}

While the size of the range $(\First_\ell, \Last_\ell)$
 is greater than one (i.e., $\First_\ell < \Last_\ell$),
all the suffixes in $SAA[\First_\ell..\Last_\ell]$ are prefixed by $P[\ell..p]$.
%
When $\First_\ell = \Last_\ell$, however,
 some suffixes in $SAA[\First_\ell]$ may not be prefixed by $P[\ell..p]$.
For example, when assuming that $P=$ {\tt aaacc},
 we have $(\First_2,\Last_2) = (5,5)$, and
 the suffixes $(1,8)$ and $(4,8)$ in $SAA[5]$ are prefixed by {\tt aacc}
 but the other suffixes $(2,8)$ and $(3,8)$ are not.
Also,  we have $(\First_1,\Last_1) = (4,4)$, and
 the suffix $(1,7)$  in $SAA[4]$ is prefixed by {\tt aaacc}
 but the suffix $(2,7)$ is not.
Thus, in addition to the range $(\First_\ell, \Last_\ell)$,
 we maintain the set $\Z_\ell$ defined as follows:
\bi
\item When $\First_\ell = \Last_\ell$,
    $\Z_\ell$ is the set of the string numbers of the suffixes
     prefixed by $P[\ell..p]$.
\ei
For simplicity, we define $\Z_\ell$ to be $\{ 1, \ldots, m \}$ when $\First_\ell < \Last_\ell$.
Then, regardless of the size of the range $(\First_\ell,\Last_\ell)$,
 a suffix $(j,q)$ is prefixed by $P[\ell..p]$
 if and only if $(j,q) \in SAA[\First_\ell..\Last_\ell]$ and $j\in \Z_\ell$.
%


\begin{algorithm} [t]
\caption{ \label{alg:FMA-search}
\textsf{BackwardSearch($P[1..p]$)} \Comment{using the FM-index of alignment}
}
\begin{algorithmic}[1]

\State $\Z \gets \{1,\ldots,m\}$; \Comment{Set of all string numbers}
\State $\sigma \gets P[p]$,
         ~$\First \gets C[\sigma]+1$, ~$\Last \gets C[\sigma+1]$,
         ~$\ell \gets p-1$;

\While{$(\First \le \Last)$ {\bf and} $\Z \neq \emptyset$ {\bf and} $(\ell \ge 1)$}
    \State $\sigma \gets P[\ell]$, ~$\First' \gets \First$, ~$\Last' \gets \Last$;
                \Comment{Previous range}
    \State $\First \gets C[\sigma] + occ(\sigma,\First-1)+1$,
           ~$\Last \gets C[\sigma] + occ(\sigma,\Last)$;

\smallskip

    \If{$\First \ge \Last$}
        \State $Z_{\rm m} \gets \{ j \,|\, (j,q) \in SAA[i]$ such that
              $\First' \le i \le \Last'$ and $B_{\sigma}[i] = 1 \}$;
         \If{$Z_{\rm m} \neq \emptyset$}
            \State $\First \gets \Last$, ~$\Z \gets \Z \cap Z_{\rm m}$;
         \Else
            \State $Z_{\rm c} \gets \{ j \,|\, (j,q) \in SAA[\First..\Last] \}$,
                    \Comment{If $\First > \Last$, $Z_{\rm c} = \emptyset$}
            \State $\Z \gets \Z \cap Z_{\rm c}$;

        \EndIf

    \EndIf
    \State $\ell \gets \ell-1$;
\EndWhile

\medskip

\ForAll{$(j,q) \in SAA[\First..\Last]$}            \Comment{If $\First > \Last$, no occurrence}
    \If{$j \in \Z$} ~~print ``$(j, q)$";           \Comment{Reporting an occurrence}
    \EndIf
\EndFor

\end{algorithmic}
\end{algorithm}

\partitle{pseudo-code}

Algorithm~\ref{alg:FMA-search} shows the search algorithm using our index,
 which is the same as the code in~\cite{tcs/NaKPLLMP16}.
(Since the definition of the a-suffix is different from that in~\cite{tcs/NaKPLLMP16},
 however, we need a correctness proof which will be given later.)
The algorithm maintains the following loop invariant
 for a range $(\First, \Last)$ and a string number set $\Z$:
\bi
\item[] At the end of Step $\ell = p,\ldots, 1$,
 the range $(\First, \Last) = (\First_\ell, \Last_\ell)$ and $\Z = \Z_\ell$.
\ei
Initially (in Step $p$), we set $(\First,\Last) = (\First_p, \Last_p)$
 and $\Z = \Z_p$ (lines 1--2).
Each iteration of the while loop (lines 3--13)
 represents each Step $\ell=p-1,\ldots,1$.
In Step $\ell=p-1,\ldots,1$,
 we first compute range $(\First,\Last)$ using the LF-mapping of
 the previous range $(\First',\Last') = (\First_{\ell+1}, \Last_{\ell+1})$
 and $\sigma = P[\ell]$ (lines 4--5).
If the size of the range $(\First,\Last)$ is more than one,
 then $(\First,\Last) = (\First_\ell,\Last_\ell)$
 and $\Z = \Z_\ell = \{1,\ldots,m\}$.
Thus, we continue to the next step (by skipping lines 6--12 and going to line 13).
Otherwise (i.e., the size of $(\First,\Last)$ is one or less),
 we compute $\Z_\ell$ as follows (lines 6--12).
Let $Z_{\rm m}$ be the set of the string numbers in $SAA[i]$'s
 such that $\First_{\ell+1}\le i \le \Last_{\ell+1}$ and $B_\sigma[i] = 1$,
and let $Z_{\rm c}$ be the set of the string numbers in $SAA[\Last]$.
Then, $\Z_\ell = \Z_{\ell+1} \cap Z_{\rm m}$ if $Z_{\rm m} \neq \emptyset$,
 and $\Z_\ell = \Z_{\ell+1} \cap Z_{\rm c}$, otherwise.
(As in~\cite{tcs/NaKPLLMP16}, lines 10--12 for $Z_{\rm c}$ can be removed
 by using a loose definition and a lazy update for $\Z_\ell$.)
For example, assume $P=$ {\tt aaacc}.
In Step 2, given $(\First_3,\Last_3) = (8,8)$ and  $\Z_{3} = \{1,2,3,4\}$,
 we have $Z_{\rm m} = \{1,4\}$ and thus $\Z_{2} = \{1,4\}$.
In Step 1, given $(\First_2,\Last_2) = (5,5)$ and  $\Z_{2} = \{1,4\}$,
 we have $Z_{\rm m} = \emptyset$ and $Z_{\rm c} = \{1,2\}$ ($\Last_1 = 4$),
 and thus  $\Z_{1} = \{1,4\} \cap \{1,2\} = \{1\}$.
After the while loop terminates,
 the occurrences of $P$ are reported using the range $(\First,\Last)$ and $\Z$ (lines 14--15).
Since the SAA stores positions in the transformed alignment $\widetilde{\Upsilon}$,
 we need to convert them to the original positions in the given strings $S^j$,
 which can be easily done by using gap information.


\partitle{correctness}

Now we show the invariant is satisfied at the end of each step (an iteration of the while loop) by induction.
Trivially, the invariant is true at the end of Step $p$, which is the induction basis.
At the beginning of Step $\ell = p-1,\ldots,1$, by inductive hypothesis,
 $(\First, \Last) = (\,\First_{\ell+1}, \Last_{\ell+1})$.
After executing line 5,
 $\,\First =  C[\sigma] + occ(\sigma,\First_{\ell+1}-1) +1$
 and $\Last =  C[\sigma] + occ(\sigma,\Last_{\ell+1})$,
 where $\sigma = P[\ell]$.
Then, the following lemmas show Algorithm~\ref{alg:FMA-search}
 computes correctly $(\,\First_{\ell}, \Last_{\ell})$ and $\Z_\ell$
 at the end of Step $\ell$.
\begin{lemma} \label{lem:range2}
If $\,\First < \Last$,
 then $(\,\First_\ell, \Last_\ell) = (\First, \Last)$
 and  $\Z_\ell = \Z_{\ell+1}$.
\end{lemma}
\begin{proof}
By definition of the LF-mapping,
 the suffixes in $SAA[\First..\Last]$
 are prefixed by $P[\ell..p]$.
We show that no suffix outside $SAA[\First..\Last]$ are prefixed by $P[\ell..p]$.
Suppose that a suffix prefixed by $P[\ell..p]$ is contained
 in an $SAA[i]$ outside $SAA[\First..\Last]$.
Then, all suffixes in $SAA[i]$ are prefixed by $P[\ell..p]$.
(If two suffixes in two distinct entries of the SAA are prefixed by $P[\ell..p]$,
 then all the suffixes in the two entries are prefixed by $P[\ell..p]$,
 which can be easily shown using the definition of the a-suffix.)
Let $(\sigma,k)$ be the smallest $\mathcal{L}_\sigma$-pair
 such that $LF(\sigma,k) = i$.
Since the suffixes in $SAA[k]$ are prefixed by $P[\ell+1..p]$,
 $k$ is included in the previous range $(\,\First_{\ell+1}, \Last_{\ell+1})$
 by definition
 and thus its LF-mapping value $i$ is also included in $(\First, \Last)$
 (note that the pair $(\sigma,k)$ is always counted in the function $occ$).
It contradicts with the assumption that $i$ is outside the range $(\First, \Last)$.
Therefore, we get $(\,\First_\ell, \Last_\ell) = (\First, \Last)$.
Furthermore,
 since $\First_\ell \neq \Last_\ell$,
 $\Z_\ell= \Z_{\ell+1} = \{1,\ldots,m\}$ by definition.
\end{proof}

\begin{lemma} \label{lem:range1m}
If $\,\First \ge \Last$ and $Z_{\rm m} \neq \emptyset$,
 then $ (\,\First_\ell, \Last_\ell) = (\Last, \Last)$
 and $\Z_\ell = \Z_{\ell+1} \cap Z_{\rm m}$.
\end{lemma}
\begin{proof}
Since $Z_{\rm m} \neq \emptyset$,
 there exist $\mathcal{L}_{\sigma}$-pairs of (m:1)-type in $(\First_{\ell+1},\Last_{\ell+1})$.
Furthermore, all of the $\mathcal{L}_{\sigma}$-pairs
 are mapped to one entry $F[\Last]$ of $F$
 since $\First \ge \Last$.
Therefore, $(\First_\ell,\Last_\ell) = (\Last,\Last)$.

Next, let us consider $\Z_\ell$.
In this case, $\mathcal{L}_{\sigma}$-pairs in $(\First_{\ell+1},\Last_{\ell+1})$
 are all of (m:1)-type.
Moreover, for every $\mathcal{L}_{\sigma}$-pair $(\sigma,i)$ in $(\First_{\ell+1},\Last_{\ell+1})$,
 the last characters of the suffixes in $SAA[i]$ are all $\sigma$ by Lemma~\ref{lem:m-1}.
Thus, $Z_{\rm m}$ is the set of the string numbers of the suffixes
 whose last characters are $\sigma$ in $SAA[\First_{\ell+1}..\Last_{\ell+1}]$.
By definition,
 $\Z_{\ell+1}$ is the set of the string numbers of the suffixes
 prefixed by $P[\ell+1..p]$.
Thus, a suffix $(j,q)$ in $SAA[\Last]$ is prefixed by $\sigma P[\ell+1..p]$ ($= P[\ell..P]$)
 if and only if $j \in Z_{\rm m}$ and $j \in \Z_{\ell+1}$.
Therefore, we get  $\Z_\ell = Z_{\rm m} \cap \Z_{\ell+1}$.
\end{proof}

\begin{lemma} \label{lem:range1c}
If $\,\First \ge \Last$ and $Z_{\rm m} = \emptyset$,
 then $(\,\First_\ell, \Last_\ell) = (\First, \Last)$
 and $\Z_\ell = \Z_{\ell+1} \cap Z_{\rm c}$.
\end{lemma}
\begin{proof}
Since $Z_{\rm m} = \emptyset$,
 there is no $\mathcal{L}_{\sigma}$-pair of (m:1)-type in $(\First_{\ell+1},\Last_{\ell+1})$.
If $\First = \Last$,
 there is one $\mathcal{L}_{\sigma}$-pair of (1:1)-type in $(\First_{\ell+1},\Last_{\ell+1})$.
If $\First > \Last$,
 there is no $\mathcal{L}_{\sigma}$-pair of (1:1)-type in $(\First_{\ell+1},\Last_{\ell+1})$.
In both cases, $(\First_\ell,\Last_\ell) = (\First,\Last)$.

Next, let us consider $\Z_\ell$ when $\First = \Last$.
Let $(\sigma, i)$
 be the only one $\mathcal{L}_{\sigma}$-pair in $(\First_{\ell+1},\Last_{\ell+1})$.
Since $(\sigma, i)$ is of (1:1)-type,
 the set of the string numbers in $SAA[\Last]$ (i.e., $Z_{\rm c}$)
 is the same as the set of the string numbers of the suffixes whose last characters are $\sigma$
 in $SAA[\First_{\ell+1}..\Last_{\ell+1}]$.
Thus, a suffix $(j,q)$ in $SAA[\Last]$ is prefixed by $\sigma P[\ell+1..p]$ ($= P[\ell..P]$)
 if and only if $j \in Z_{\rm c}$ and $j \in \Z_{\ell+1}$.
Therefore, we get  $\Z_\ell = Z_{\rm c} \cap \Z_{\ell+1}$.
\end{proof}
Therefore, we can get the following theorem.
\begin{theorem} \label{lem:search-cor}
Algorithm~\ref{alg:FMA-search} finds correctly all the occurrences of a pattern $P$.
\end{theorem}


\subsection{Data structures}
\label{subsec:ds}


Our index consists of the function $occ$, the array $C$, the bit-vectors $B_\sigma$,
 and a sampled SAA.
Furthermore, we store gap information for mutual conversion
 between positions in an original string $S^j$
 and positions in its transformed string $\widetilde{S}^j$.

\partitle{sampled SAA}

We store the SAA using two kinds of sampling as in~\cite{tcs/NaKPLLMP16},
 the {\em regular-position sampling} and
 the {\em irregular-position sampling}.
For the regular-position sampling,
 we sample $SAA[i]$ storing every $d$-th position
  in the transformed alignment $\widetilde{\Upsilon}$
 where $d$ is a given parameter.
Then, we get an $SAA[i]$ in a sampled SAA by  repeating the LF-mapping from $SAA[i]$
 until a sampled entry $SAA[k]$ is encountered.
In order to guarantee that
 the string numbers in $SAA[i]$ are the same as the ones in $SAA[k]$,
 we also need the following irregular sampling:
 an $SAA[i]$ is sampled when $L[i]$ has multiple characters or,
 for any $\sigma \in \Sigma$, the pair $(\sigma,i)$ is of (m:1)-type.
Note that such an $SAA[i]$ has different string numbers
 from the string number in $SAA[i']$ where $i' = LF(\sigma,i)$ for a character $\sigma \in L[i]$.

\partitle{sampled ISAA}

For supporting retrieval operations,
 we also need a sampled inverse SAA.
In the FMA without gaps~\cite{tcs/NaKPLLMP16},
 the regular-position sampling is enough for the inverse SAA.
Due to gaps, however,
 we need also an irregular sampling for the inverse SAA.
Suppose that a gap in a transformed string $\widetilde{S}^j$
 includes a regular sampling position $q$.
Then, we cannot sample the position $q$ in $\widetilde{S}^j$.
Let $q'$ be the leftmost position such that
 $q' > q$ and no gap in $\widetilde{S}^j$ includes $q'$.
Then, the position $q'$, instead of $q$, is sampled in $\widetilde{S}^j$.
For example, assuming position 4 is a regular sampling position in Fig.~\ref{fig:alignment},
 instead of position 4, position 6 is sampled in $\widetilde{S}^4$
 since $\widetilde{S}^4[3..5]$ is a gap.


\section{Experiments}
\label{sec:exp}


To compare our FM-index of alignment (FMA) with the RLCSA~\cite{jcb/MakinenNSV10},
 we measured their sizes, pattern search time and retrieval time.
Our index was implemented using SDSL (Succinct Data Structure Library~\cite{sea/GogBMP14})
and all experiments were conducted on a computer with Intel Xeon X5672 CPU and 32GB RAM,
 running the Linux debian 3.2.0-4-amd64 operating system.

The experimental data set is a reference sequence and 99 individual sequences,
 which are downloaded from the 1000 Genomes Project website.
The reference genome is hs37d5 and each individual sequence consists of a pair of BAM and BAI files,
 where a BAM file contains reads (short segments of length 90-125) of each individual
 and a BAI file contains the alignment of the reads.
Each pair of BAM and BAI files is fed to SAMtools (Sequence Alignment/Map tools) to obtain a VCF file
 which stores genetic mutations such as substitutions, insertions and deletions relative to the reference genome.


\begin{table}[t]
\centering \small
\caption{\label{tbl:exp-space}
The index sizes (in MBytes) where
 ``sampling" means the space for sampling,
 ``gap" means the space for storing gap information,
  and ``core" means the space except for sampling and gap.
}
\begin{tabular}{c c | r r r | r r r}
\hline 
\multicolumn{2}{c|}{Number of sequences} & \multicolumn{3}{c|}{30} & \multicolumn{3}{c}{100} \\
\hline 
\multicolumn{2}{c|}{Sampling rate}
    & \multicolumn{1}{c}{32}    & \multicolumn{1}{c}{128}   & \multicolumn{1}{c|}{512}
    & \multicolumn{1}{c}{32}    & \multicolumn{1}{c}{128}   & \multicolumn{1}{c}{512} \\
\hline 
 & total & 57.5 & 49.2 & 47.1 & 85.5 & 75.6 & 73.1 \\
FMA & core & 38.2 & 38.2 & 38.2 & 49.7 & 49.7 & 49.7 \\
 & gap & 1.4 & 1.4 & 1.4 & 5.1 & 5.1 & 5.1 \\
 & sampling & 17.9 & 9.6 & 7.5 & 30.7 & 20.8 & 18.3 \\

\hline 
 & total & 390 & 193 & 141 & 1113 & 417 & 233 \\
RLCSA & core & 122 & 122 & 122 & 168 & 168 & 168 \\
 & sampling & 268 & 71 & 19 & 945 & 249 & 65 \\

\hline 
\end{tabular}

\end{table}


First, we compared the sizes of our FMA with those of RLCSA.
We created these indexes with sampling rates d = 32, 128, and 512 from 30 and 100 sequences
 (Table~\ref{tbl:exp-space}).
The table shows that the FMA requires less than one third of the space of RLCSA in every case.
Furthermore, the size of the FMA varies little regardless of the sampling rates
 because irregular sampling occupies most space for sampling.


\begin{table}[t]
\centering \small
\caption{\label{tbl:exp-loc}
Pattern search (location) time (in secs) for 500 queries of length 10.
}
\begin{tabular}{c | r r r | r r r}
\hline 
Number of sequences & \multicolumn{3}{c|}{30} & \multicolumn{3}{c}{100} \\
\hline 
Sampling rate& \multicolumn{1}{c}{32} & \multicolumn{1}{c}{128} &   \multicolumn{1}{c|}{512}
    & \multicolumn{1}{c}{32}    & \multicolumn{1}{c}{128}   & \multicolumn{1}{c}{512} \\
\hline 
FMA & 6.40 & 19.86 & 48.63 & 15.13 & 25.53 & 37.84 \\
RLCSA & 6.94 & 31.08 & 177.89 & 20.04 & 102.36 & 622.42 \\
\hline 
\end{tabular}

\end{table}


\begin{figure}[t]

{\centering
    \includegraphics[height=4.0cm]{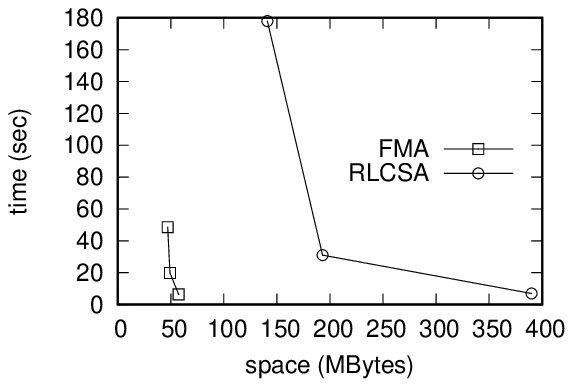}
    \quad
    \includegraphics[height=4.0cm]{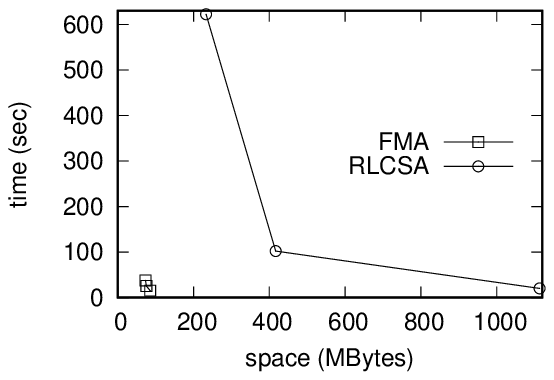}\\
   (a) 30 sequences
   \hspace{4cm}
   (b) 100 sequences\\
}
\caption{\label{fig:exp-loc}
Total index sizes and pattern search (location) times for 500 queries of length 10.
Each index was tested with sampling rates $d = 32$, 128, and 512.
}
\end{figure}


Second, we compared the running time of pattern search (location) reporting all occurrences.
We performed the pattern search with patterns of length 10 on the indexes
 with sampling rates d = 32, 128, and 512 from 30 and 100 sequences
 (Table~\ref{tbl:exp-loc} and Fig.~\ref{fig:exp-loc}),
 and FMA is the fastest in all cases.
We also compared the retrieval time
 (Table~\ref{tbl:exp-ret} and Fig.~\ref{fig:exp-ret}).
In all cases, RLCSA shows the best performance in retrieval time.


\begin{table}[t]
\centering \small
\caption{\label{tbl:exp-ret}
Retrieval time (in secs) for 500 queries of length 10.
}
\begin{tabular}{c | r r r | r r r}
\hline 
Number of sequences & \multicolumn{3}{c|}{30} & \multicolumn{3}{c}{100} \\
\hline 
Sampling rate& \multicolumn{1}{c}{32} & \multicolumn{1}{c}{128} &   \multicolumn{1}{c|}{512}
    & \multicolumn{1}{c}{32}    & \multicolumn{1}{c}{128}   & \multicolumn{1}{c}{512} \\
\hline 
FMA & 0.05 & 0.11 & 0.40 & 0.04 & 0.12 & 0.43 \\
RLCSA & 0.01 & 0.02 & 0.10 & 0.01 & 0.02 & 0.10 \\
\hline 
\end{tabular}

\end{table}



\begin{figure}[t]

{\centering
    \includegraphics[height=4.0cm]{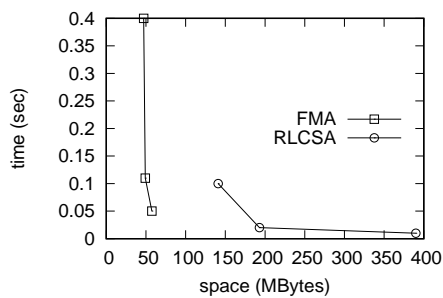}
    \quad
    \includegraphics[height=4.0cm]{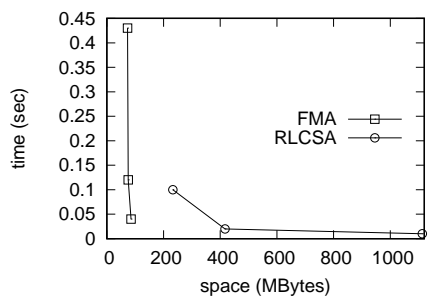}\\
   (a) 30 sequences
   \hspace{4cm}
   (b) 100 sequences\\
}
\caption{\label{fig:exp-ret}
Total index sizes and retrieval times for 500 queries of length 10.
Each index was tested with sampling rates $d = 32$, 128, and 512.
}
\end{figure}



\section{Concluding Remarks}
\label{sec:conclusion}


We have proposed the FM-index of alignment with gaps,
 a realistic index for similar strings, which allows gaps in their alignment.
For this,
 we have designed a new version of suffix array of alignment
 by using alignment transformation and a new definition of the alignment-suffix.
The new SAA enabled us to support the LF-mapping and backward search
 regardless of gap existence in alignments.
Experimental results showed that our index is more space-efficient than RLCSA
 and it is faster than RLCSA in pattern search
even though its retrieval time is slower than that of RLCSA.
It remains as future work to do extensive experiments and analysis on various real-world data.


\section*{Acknowledgements}
\small

Joong Chae Na was supported
  by the MISP(Ministry of Science, ICT \& Future Planning), Korea,
  under  National program for Excellence in Software program (the SW oriented college support grogram) (R7718-16-1005)
  supervised by the IITP (Institute for Information \& communications Technology Promotion),
  and by Basic Science Research Program
    through the National Research Foundation of Korea (NRF)
    funded by the Ministry of Science, ICT \& Future Planning (2014R1A1A1004901).
%
Heejin Park was supported
  by the research fund of Signal Intelligence Research Center supervised
    by Defense Acquisition Program Administration and Agency for Defense Development of Korea.
Thierry Lecroq, Martine L\'eonard  
 and Laurent Mouchard were supported
  by the French Ministry of Foreign Affairs Grant 27828RG (INDIGEN, PHC STAR 2012).
%
Kunsoo Park was supported
  by the Bio \& Medical Technology Development Program of the NRF
  funded by the Korean government, MSIP (NRF-2014M3C9A3063541).

\bigskip

%
%




\end{document}